\documentclass[
twocolumn,
showpacs,
preprintnumbers,
aps,
prd,
a4paper,
superscriptaddress,
nofootinbib,
tightenlines,
floats,floatfix
]{revtex4-2}

\usepackage{graphicx}
\usepackage{dcolumn}
\usepackage{bm}
\usepackage{amsmath,amssymb}
\usepackage{float}
\usepackage{array,booktabs,makecell}
\usepackage{hyperref}

\usepackage{subfigure}   

\usepackage{amsthm}
\newtheorem{theorem}{Theorem}

\usepackage[normalem]{ulem}  
\usepackage{soul}             

\begin{document}

\title{Stability Protected Phantom Bound in Expansion Modulated Cosmology}

\author{Prasanta Sahoo}
\email{prasantmath123@yahoo.com}
\affiliation{Midnapore College (Autonomous), Midnapore, West Bengal 721101, India}

\date{\today}

\begin{abstract}
Recent cosmological observations, including DESI Data Release 2 (DR2) \cite{DESIDR2}, allow for mild redshift evolution of the dark energy equation of state (EoS), motivating renewed interest in the phantom regime ($w<-1$). A no-go result is presented for a class of single field effective scalar cosmologies with Hubble modulated kinetic response, as motivated by infrared modified and nonlocal gravitational frameworks \cite{DeserWoodard2007,Maggiore2014}. Imposing ghost freedom, $\mathcal{M}\equiv \partial\rho_\phi/\partial X>0$, renders the phantom divide ($w_\phi=-1$) an invariant and dynamically stable manifold of the cosmological flow, extending standard kinematical no-go theorems to a dynamical systems framework. Continuous ghost free evolution into the phantom regime is forbidden, although $w_\phi\to -1$ can be approached asymptotically. The late time dynamics generically converge to a de~Sitter like attractor driven by expansion induced kinetic suppression rather than potential fine tuning. These results clarify the role of stability constraints in shaping late time cosmic acceleration in single field effective scalar cosmologies.
\end{abstract}

\maketitle

\let\oldthefootnote\thefootnote
\renewcommand{\thefootnote}{\textit{}}
\footnotetext{Corresponding author: \href{mailto:prasantmath123@yahoo.com}{prasantmath123@yahoo.com}}
\let\thefootnote\oldthefootnote

\section{Introduction}

Late time cosmic acceleration is supported by multiple independent observations, including Type~Ia supernovae, cosmic microwave background (CMB) anisotropies, and baryon acoustic oscillations (BAO) \cite{Riess1998,Perlmutter1999,Planck2018,Eisenstein2005}. The concordance $\Lambda$CDM model, in which cosmic acceleration is driven by a cosmological constant, remains consistent with current datasets. Nevertheless, the extreme fine tuning required to match the observed vacuum energy density and the coincidence problem motivate the exploration of dynamical dark energy scenarios \cite{Weinberg1989,Zlatev1999,Padmanabhan2003,Peebles2003,Carroll2001}.

Recent results from the Dark Energy Spectroscopic Instrument (DESI) Data Release 2 (DR2) \cite{DESIDR2}, combined with other late time probes, allow for a mild redshift evolution of the dark energy EoS parameter $w(z)$. This raises the question of whether cosmological dynamics can naturally approach $w\simeq -1$ within a theoretically consistent and stable framework, rather than requiring a fundamental cosmological constant \cite{Copeland2006,Tsujikawa2013}.

Scalar field models provide a minimal and well motivated framework for dynamical dark energy. Canonical quintessence models predict $w\geq -1$ and are free from ghost instabilities \cite{Ratra1988,Wetterich1988}. In contrast, the phantom regime ($w<-1$) typically requires a wrong sign kinetic term, leading to catastrophic vacuum instabilities \cite{Caldwell2002,Carroll2003}. Noncanonical $k$-essence theories extend the phenomenology by allowing nonlinear kinetic terms and attractor behavior \cite{ArmendarizPicon2001,Chiba2000,Garriga1999}. However, it is well known that continuous, ghost free crossing of the phantom divide is forbidden in broad classes of single field theories with local Lagrangians \cite{Vikman2005,Nojiri2005}. Phantom crossing generally requires additional propagating degrees of freedom, higher derivative operators, or non-minimal couplings \cite{Creminelli2009,Hu2004}.

In this work, a class of single field effective scalar cosmologies is investigated in which the kinetic response is modulated by the Hubble expansion rate. Such Hubble dependent kinetic structures arise naturally in infrared modified and nonlocal gravitational frameworks, where curvature dependent operators renormalize the effective matter sector dynamics on cosmological backgrounds \cite{DeserWoodard2007,Maggiore2014,Woodard2014}. The model is treated as an effective background level parametrization, capturing the infrared behavior of curvature modulated scalar dynamics without specifying a unique ultraviolet completion.

A central result of this work is a dynamical no-go theorem for phantom crossing in this class of models. Imposing ghost freedom, defined by the positivity of the kinetic response function $\mathcal{M}\equiv \partial\rho_\phi/\partial X>0$, renders the phantom divide $w_\phi=-1$ an invariant manifold of the background cosmological flow. Unlike standard kinematical no-go arguments, the phantom divide is shown to be dynamically selected as an asymptotically stable attractor in phase space. The late time dynamics generically converge toward a de~Sitter like attractor with $w_\phi\to -1$, driven by expansion induced kinetic suppression rather than finetuning of the scalar potential.

While phantom crossing remains possible in multi-field theories and in higher derivative or non-minimally coupled scalar tensor frameworks such as Horndeski, beyond Horndeski, and DHOST theories \cite{Deffayet2010,Gleyzes2015,Langlois2016,Bellini2014}, it is prohibited in the class of single field effective descriptions considered here under the assumptions of background energy conservation and ghost freedom. This framework therefore provides a dynamical interpretation of the observed proximity of dark energy to a cosmological constant and clarifies the role of stability constraints in shaping late time cosmic acceleration.

The paper is organized as follows. Section~\ref{sec:effective_framework} introduces the effective Hubble modulated scalar framework and derives the background equations and stability conditions. Section~\ref{sec:asymptotic_dynamics} analyzes the asymptotic late time dynamics and de~Sitter attractor behavior. Section~\ref{sec:covariant_origin} discusses the possible covariant origin of the effective model from nonlocal gravity. Section~\ref{sec:Background_Reformulation} reformulates the single field phantom no-go condition in a background dynamical framework and establishes the stability protected phantom bound. Section~\ref{sec:observational_illustration} provides an observational illustration of the background expansion history. Section~\ref{sec:perturbations} comments on perturbative stability and limitations of the effective description. Finally, Section~\ref{sec:conclusion} summarizes the results and outlines future directions.

\section{Effective Framework and Stability Bound}\label{sec:effective_framework}

Late time acceleration is analyzed at the homogeneous background level within scalar field cosmologies featuring a Hubble modulated kinetic response, while the gravitational sector is described by Einstein gravity. Deviations from $\Lambda$CDM arise solely from the scalar sector. The framework assumes a single effective homogeneous scalar degree of freedom. The analysis is restricted to spatially flat Friedmann--Lemaître--Robertson--Walker (FLRW) spacetimes \cite{Mukhanov2005}, with a homogeneous scalar field $\phi(t)$ and kinetic variable
\begin{equation}
X \equiv \frac{\dot{\phi}^{2}}{2}.
\end{equation}

The scalar sector is described by the effective background pressure
\begin{equation}
P_{\phi}(X,H,\phi)
=
X - V(\phi)
- \beta\,\frac{X^{2}}{H^{2}+X},
\label{eq:Peff}
\end{equation}
where $H=\dot a/a$ and $\beta>0$ parametrizes expansion induced suppression of the kinetic response. The parameter $\beta$ has mass dimension two and may be expressed as $\beta=\tilde{\beta}H_0^2$ with dimensionless $\tilde{\beta}$. Eq.~(\ref{eq:Peff}) is interpreted as an effective background level parametrization rather than a fundamental covariant Lagrangian. The explicit Hubble dependence is understood as arising after integrating out gravitational and nonlocal curvature degrees of freedom in an infrared modified effective description. The framework introduces no additional propagating degrees of freedom at the background level.

The corresponding energy density, defined by $\rho_{\phi}=2X P_{\phi,X}-P_{\phi}$, is
\begin{equation}
\rho_{\phi}
=
X + V(\phi)
- \beta\,\frac{X^{2}(3H^{2}+X)}{(H^{2}+X)^{2}} .
\label{eq:rhoeff}
\end{equation}

Background evolution satisfies the continuity equation $\dot{\rho}_{\phi}+3H(\rho_{\phi}+P_{\phi})=0$. Ghost freedom is imposed through positivity of the kinetic response function
\begin{equation}
\mathcal{M}
\equiv
\frac{\partial \rho_{\phi}}{\partial X}
=
1
- \beta\,\frac{X\left(6H^{4}+3H^{2}X+X^{2}\right)}{(H^{2}+X)^{3}},
\label{eq:ghost}
\end{equation}
which restricts the admissible phase space. The imposition of $\mathcal{M} > 0$ establishes a stability protected boundary in the cosmological phase space. Using the identity $\rho_{\phi} + P_{\phi} = 2X P_{\phi,X}$, the ghost free condition $P_{\phi,X} \geq 0$ acts as a constraint preventing the EoS parameter from entering the phantom regime in single field effective descriptions. For any nonvanishing kinetic energy $X > 0$, positivity of the kinetic response ensures that $\rho_{\phi} + P_{\phi} \geq 0$, rendering the phantom divide $w_{\phi} = -1$ a boundary of the physically admissible domain. Consequently, any continuous transition into $w_{\phi} < -1$ would require $P_{\phi,X} < 0$, signaling the onset of ghost instabilities and the breakdown of the effective single field description.

Within the ghost free domain, the combination $\rho_{\phi}+P_{\phi}=2X P_{\phi,X}\ge 0$, with equality attained only in the limit $X\to 0$. Consequently, the EoS parameter $w_{\phi}=P_{\phi}/\rho_{\phi}$ satisfies $w_{\phi}\ge -1$, identifying the phantom divide as a stability boundary.

\begin{theorem}[Stability protected phantom bound]
\label{thm:phantom_nogo}
Consider a spatially flat FLRW cosmology governed by Einstein gravity and containing a single homogeneous scalar degree of freedom with effective pressure $P_\phi=P(X,H,\phi)$. Assume:  
(i) background energy conservation;  
(ii) ghost free kinetic response, $\partial P_\phi/\partial X \ge 0$ and $\partial\rho_\phi/\partial X>0$;  
(iii) expansion induced suppression of the kinetic response, such that $\partial P_\phi/\partial X \to 0$ in the late time attractor regime for finite $X$.  

Then all ghost free cosmological solutions satisfy $w_\phi \ge -1$, with equality attained only asymptotically as $X\,\partial P_\phi/\partial X \to 0$. No continuous ghost free cosmological evolution can cross the phantom divide. Under assumption (iii), cosmological evolution drives $w_\phi \to -1$ as $t\to\infty$.
\end{theorem}

\begin{proof}
The homogeneous background dynamics define an autonomous system $\dot{\mathbf{u}}=\mathbf{F}(\mathbf{u})$ with $\mathbf{u}=(X,\phi,H)$. Define the hypersurface function
\begin{equation}
\Sigma \equiv \rho_\phi + P_\phi = 2X P_{\phi,X}.
\end{equation}
The ghost free domain is characterized by $P_{\phi,X}\ge0$ and $\partial\rho_\phi/\partial X>0$. The hypersurface
\begin{equation}
\mathcal{S}:\quad \Sigma=0 \quad \Leftrightarrow \quad P_{\phi,X}=0
\end{equation}
corresponds to the phantom divide $w_\phi=-1$ for finite $X$.

A hypersurface $\mathcal{S}$ is invariant if the vector field is tangent to it, i.e.
\begin{equation}
\left.\frac{d\Sigma}{dt}\right|_{\mathcal{S}} = 0 .
\end{equation}
Using $\Sigma=2X P_{\phi,X}$, one obtains
\begin{equation}
\dot{\Sigma} = 2\dot{X} P_{\phi,X} + 2X \dot{P}_{\phi,X},
\end{equation}
and therefore
\begin{equation}
\left.\dot{\Sigma}\right|_{\mathcal{S}} = 2X\,\left.\dot{P}_{\phi,X}\right|_{P_{\phi,X}=0}.
\end{equation}

Assumption (iii) implies that $P_{\phi,X}(X,H,\phi)$ is continuously differentiable and that its dominant time dependence near the attractor is controlled by the Hubble rate. Linearizing the evolution of $P_{\phi,X}$ in a neighborhood of $\mathcal{S}$ yields
\begin{equation}
\dot{P}_{\phi,X} = -\lambda(X,\phi,H)\,P_{\phi,X} + \mathcal{O}(P_{\phi,X}^{2}),
\end{equation}
where $\lambda(X,\phi,H)>0$ for expanding solutions. For the explicit model in Eq.~(\ref{eq:Peff}), one has $\partial P_{\phi,X}/\partial H < 0$ for $H>0$ and $X>0$, implying $\lambda>0$ in the accelerating regime. Consequently,
\begin{equation}
\left.\dot{P}_{\phi,X}\right|_{P_{\phi,X}=0}=0,
\end{equation}
and hence $\left.\dot{\Sigma}\right|_{\mathcal{S}} = 0$, demonstrating that $\mathcal{S}$ is an invariant manifold of the cosmological flow.

Linearization around $\mathcal{S}$ gives
\begin{equation}
\delta \dot{P}_{\phi,X} = -\lambda\,\delta P_{\phi,X},
\end{equation}
implying exponential decay of perturbations toward $\mathcal{S}$. Since $P_{\phi,X}\ge0$ in the admissible phase space, trajectories cannot cross into $P_{\phi,X}<0$, implying $w_\phi\ge -1$ and $w_\phi\to -1$ asymptotically. The phantom divide therefore corresponds to a dynamically stable invariant manifold of the cosmological flow.
\end{proof}

The monotonic Hubble dependence of $P_{\phi,X}$ is realized explicitly in Hubble modulated kinetic effective theories and is interpreted as a structural property of the effective background description, which may arise after integrating out gravitational degrees of freedom in infrared modified frameworks.

\section{Asymptotic Dynamics}
\label{sec:asymptotic_dynamics}

The late time behavior follows from the Hubble modulated kinetic structure of the effective scalar sector. The background evolution is controlled by the hierarchy between the Hubble scale and the scalar kinetic scale.

At early times, when $H^{2}\gg X$, the expansion induced correction is subdominant and the effective pressure reduces to
\begin{equation}
P_{\phi}\simeq X - V(\phi) - \beta\,\frac{X^{2}}{H^{2}},
\end{equation}
so that the scalar field behaves approximately as a canonical quintessence field with small higher order kinetic corrections. In this regime, the scalar EoS tracks the dominant background component, yielding standard radiation  and matter dominated expansion histories \cite{Copeland1998,Tsujikawa2013}.

As the expansion proceeds and the Hubble rate decreases, the Hubble modulated kinetic term suppresses the kinetic contribution, dynamically driving the system toward $X\ll H^{2}$. In this limit, one obtains
\begin{equation}
P_{\phi}\to -V(\phi), 
\qquad
\rho_{\phi}\to V(\phi),
\end{equation}
implying
\begin{equation}
w_{\phi} \to -1 .
\end{equation}
This corresponds to a dynamically approached de~Sitter like attractor solution, which arises without requiring fine tuning of the scalar potential.

During the transition regime, $w_{\phi}$ can approach the phantom divide. However, Theorem~\ref{thm:phantom_nogo} enforces $\rho_{\phi}+P_{\phi}=2X P_{\phi,X}\ge0$ in the ghost free domain, which forbids any continuous crossing of $w_{\phi}=-1$. The phantom divide therefore constitutes a dynamically invariant stability boundary, and late time acceleration arises from kinetic suppression driving the system toward, but not beyond, the de~Sitter attractor.

\section{Covariant Origin from Nonlocal Gravity}\label{sec:covariant_origin}

The Hubble modulated suppression of the scalar kinetic response can be interpreted as an infrared effective manifestation of generally covariant nonlocal gravitational dynamics \cite{DeserWoodard2007,Maggiore2014}. In such frameworks, curvature dependent operators of the schematic form $R\,f(\Box^{-1}R)$ modify the effective gravitational and matter sector dynamics on cosmological backgrounds. On homogeneous FLRW spacetimes, one typically finds $\Box^{-1}R \sim H^{-2}$ at the background level, implying that nonlocal curvature effects track the cosmic expansion rate.

A scalar field coupled to nonlocal curvature through an interaction of the form
\begin{equation}
S_\phi = -\frac{1}{2}\int d^4x \sqrt{-g}\, K(\Box^{-1}R)\,(\nabla\phi)^2 
- \int d^4x \sqrt{-g}\, V(\phi)
\end{equation}
acquires an effective kinetic response that depends on the background expansion history. Expanding around a late time FLRW background and motivates an effective background pressure of the form
\begin{equation}
P_\phi(X,H,\phi) = X - V(\phi) - \beta\,\frac{X^2}{H^2 + X},
\end{equation}
which reproduces the phenomenological structure adopted in Eq.~(\ref{eq:Peff}).

The explicit dependence on the Hubble rate is therefore interpreted as a background level effective parametrization arising from curvature induced renormalization of the scalar kinetic sector. The present analysis is restricted to homogeneous background dynamics and assumes a single effective scalar degree of freedom. The stability results derived in this work do not rely on the detailed ultraviolet completion of the nonlocal theory and follow solely from background energy conservation and positivity of the kinetic response. Consequently, the phantom divide emerges as a stability protected boundary of the cosmological dynamics. A fully covariant effective field theory completion and perturbative analysis are left for future investigation.

\section{Single-Field Phantom Bound from Background Dynamics}\label{sec:Background_Reformulation}

Any continuous phantom crossing in a single scalar cosmology governed by Einstein gravity requires either additional propagating degrees of freedom, higher derivative operators, or violation of ghost free conditions. Theorem~\ref{thm:phantom_nogo} shows that Hubble modulated kinetic suppression dynamically enforces the phantom divide as an invariant boundary of the background flow. To disentangle this dynamical mechanism from purely kinematical stability constraints, it is useful to reformulate the standard single field phantom no-go result in a model independent background framework. The following theorem provides such a reformulation for generic effective single field descriptions without explicit Hubble dependence.

\begin{theorem}[Background single field phantom no-go condition]
\label{thm:universal_stability_PRL}
Consider a spatially homogeneous single scalar degree of freedom minimally coupled to Einstein gravity, with background dynamics described by an effective pressure $P_\phi(X,\phi)$. Assume:  
(i) background energy conservation;  
(ii) ghost free kinetic response, $\partial P_\phi/\partial X \ge 0$ and $\partial \rho_\phi/\partial X > 0$;  
and (iii) absence of additional propagating degrees of freedom beyond the scalar mode.  

Then any continuous cosmological evolution satisfies $w_\phi \equiv P_\phi/\rho_\phi \ge -1$, with equality attained only if $X\,\partial P_\phi/\partial X = 0$. A continuous crossing into the phantom regime necessarily violates condition (ii), signaling either the onset of a ghost instability or the breakdown of the effective single field description.
\end{theorem}

\begin{proof}
For a spatially flat FLRW background, the Raychaudhuri equation implies
\begin{equation}
\dot H = -\frac{1}{2}(\rho_\phi + P_\phi),
\end{equation}
where units $8\pi G=1$ are assumed. Using the identity $\rho_\phi + P_\phi = 2X P_{\phi,X}$ and assumption (ii), one obtains
\begin{equation}
\dot H = - X P_{\phi,X} \le 0
\end{equation}
for all ghost free configurations. A continuous transition into a phantom regime requires $\dot H > 0$, which implies $P_{\phi,X} < 0$ and thus violates the ghost free condition. Therefore, continuous ghost free cosmological evolution cannot cross the phantom divide $w_\phi=-1$ within a single field effective description.
\end{proof}

This result reproduces the standard single field no-go theorem for phantom crossing \cite{Vikman2005}, but is derived here using a background dynamical systems formulation rather than perturbative effective field theory arguments. The reformulation allows the dynamical mechanism of Theorem~\ref{thm:phantom_nogo} to be interpreted as a specific realization of the general single field stability constraint.

\section{Observational Illustration}\label{sec:observational_illustration}

The background expansion history implied by the stability protected framework is compared with late time distance measurements from the Pantheon+ Type~Ia supernova compilation \cite{Scolnic2022}. The effective dark energy EoS is identified with $w_\phi$, and the background equations are solved numerically for the Hubble expansion rate and scalar sector, imposing the ghost free condition $\mathcal{M}>0$ throughout the evolution.

\begin{figure}[t]
\centering
\includegraphics[width=\columnwidth]{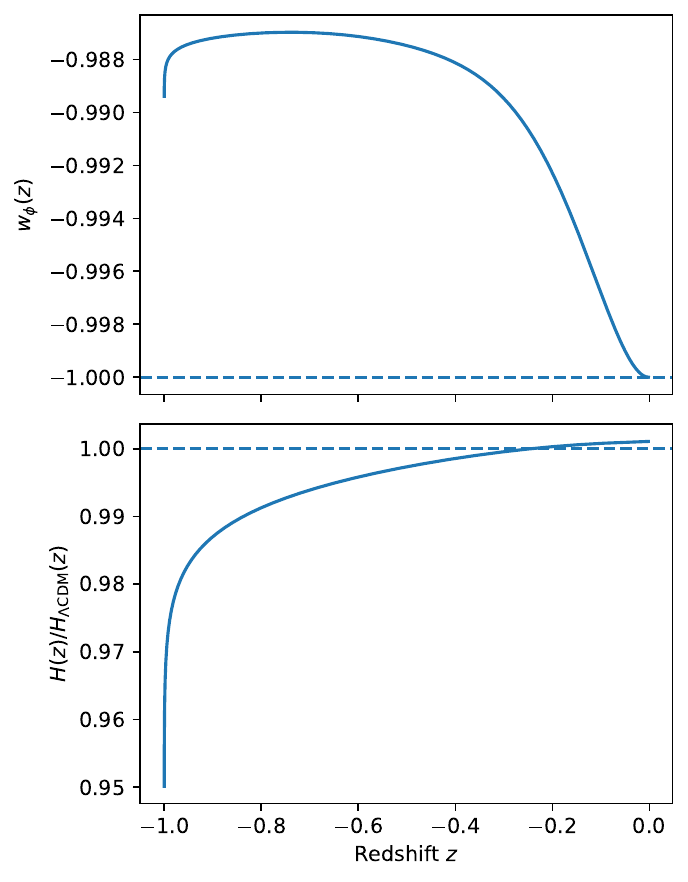}
\caption{Redshift evolution of the EoS $w_{\phi}(z)$ (top panel) and the ratio $H(z)/H_{\Lambda{\rm CDM}}(z)$ (bottom panel). Dashed lines indicate the $\Lambda$CDM reference values. The EoS approaches $w_{\phi}=-1$ at low redshift without crossing, while the expansion history remains close to the $\Lambda$CDM prediction.}
\label{fig:weff_hratio}
\end{figure}

The resulting expansion histories remain close to $\Lambda$CDM while allowing mild redshift evolution of $w_{\phi}$. Figure~\ref{fig:weff_hratio} shows convergence toward $w_{\phi}=-1$ at low redshift without phantom crossing, with $H(z)$ differing only marginally from the $\Lambda$CDM background. The kinetic response remains positive throughout the evolution, indicating that the absence of phantom crossing follows from the stability conditions of the effective framework rather than being imposed by observational priors.

\section{Linear Perturbations and Stability}\label{sec:perturbations}

Perturbative consistency requires the absence of ghost and gradient instabilities. Ghost freedom is ensured by the positivity of the kinetic response function,
\begin{equation}
\mathcal{M} \equiv \frac{\partial \rho_{\phi}}{\partial X} > 0,
\end{equation}
as defined in Eq.~(\ref{eq:ghost}).

For the Hubble modulated effective pressure $P_{\phi}(X,H,\phi)$ defined in Eq.~(\ref{eq:Peff}), scalar perturbations are characterized at the background level by an effective sound speed
\begin{equation}
c_s^2 = \frac{P_{\phi,X}}{P_{\phi,X} + 2X P_{\phi,XX}},
\end{equation}
where
\begin{equation}
P_{\phi,X}
= 1 - \beta\,\frac{2XH^{2}+X^{2}}{(H^{2}+X)^{2}},
\qquad
P_{\phi,XX}
= -2\beta\,\frac{H^{4}}{(H^{2}+X)^{3}} .
\end{equation}
Gradient stability requires $c_s^2>0$, which is satisfied within the ghost free phase space for appropriate parameter choices. In the late time attractor regime $X\to 0$ and $w_{\phi}\to -1$, one finds $P_{\phi,X}\to 1$ and $P_{\phi,XX}\to -2\beta/H^{2}$, implying $c_s^2\to 1$ .

Since the effective pressure depends explicitly on the background expansion rate $H$, the above sound speed should be interpreted as a background level effective propagation speed. A fully covariant perturbative treatment, including construction of the quadratic action for scalar fluctuations and verification of the kinetic and gradient coefficients, is beyond the scope of the present work and left for future investigation.

On sub-Hubble scales and assuming that dark energy perturbations are negligible, matter density perturbations obey
\begin{equation}
\ddot{\delta}_m + 2H\dot{\delta}_m - 4\pi G \rho_m \delta_m = 0,
\end{equation}
so that deviations from $\Lambda$CDM arise solely through the modified background expansion history. The effects of clustering dark energy and metric perturbations are deferred to future work.

\section{Conclusion}\label{sec:conclusion}

This work investigated late time cosmic acceleration in scalar field cosmologies with a Hubble modulated kinetic response within the framework of Einstein gravity. The analysis was performed at the homogeneous background level, assuming a single effective scalar degree of freedom and standard Friedmann dynamics. The Hubble dependent kinetic suppression provides a mechanism for dynamically reducing the scalar kinetic energy at late times, leading to accelerated expansion without introducing additional propagating degrees of freedom or modifying the gravitational sector.

A central result is the establishment of a stability protected phantom bound. By imposing ghost freedom, encoded in the positivity of the kinetic response function, the phantom divide $w_\phi=-1$ was shown to correspond to an invariant manifold of the background dynamical system. Within the ghost free phase space, the combination $\rho_\phi + P_\phi = 2X P_{\phi,X}$ remains non-negative, implying that continuous cosmological evolution cannot cross into the phantom regime. The dynamics instead drive the system asymptotically toward a de~Sitter like attractor with $w_\phi \to -1$, without requiring fine tuning of the scalar potential. This provides a dynamical realization of the well known single field phantom no-go theorem using a background invariant manifold formulation.

The model was further interpreted as an effective infrared parametrization of possible nonlocal gravitational dynamics, where curvature induced renormalization of the scalar kinetic sector naturally introduces Hubble dependent kinetic suppression. While a fully covariant derivation from a fundamental nonlocal theory was not constructed, the effective framework captures the qualitative late time behavior expected from curvature modulated scalar dynamics.

An observational illustration using late time expansion histories demonstrates that the model can closely mimic $\Lambda$CDM while allowing mild redshift evolution of the EoS. The convergence toward $w_\phi=-1$ at low redshift arises dynamically from stability conditions rather than being imposed phenomenologically. This suggests that the observed proximity of dark energy to a cosmological constant, as indicated by recent datasets such as DESI DR2, is consistent with stability driven scalar dynamics rather than necessarily requiring phantom physics. Phantom crossing remains possible in multi-field models or in theories with higher derivative or non-minimal couplings, but it is excluded in the class of single field effective descriptions considered here under the assumptions of background energy conservation and ghost freedom.

The perturbative analysis was restricted to a background level effective description. The explicit Hubble dependence of the scalar pressure implies that a fully covariant perturbation theory requires constructing the quadratic action for fluctuations in an effective field theory framework. The effective sound speed discussed here should therefore be regarded as a phenomenological diagnostic rather than a prediction of a fundamental covariant scalar tensor theory. A systematic treatment of scalar and metric perturbations, including clustering dark energy and implications for structure formation, is left for future work.

Several directions for further investigation remain open. A covariant effective field theory completion that reproduces the Hubble modulated kinetic structure at the background level would provide a fundamental theoretical underpinning of the model. A full parameter inference using supernova, BAO, and CMB data, including Bayesian model comparison with $\Lambda$CDM, would quantify the observational viability of the framework. Finally, the impact of kinetic suppression on cosmic growth, weak lensing, and the $H_0$ and $S_8$ tensions warrants detailed exploration.

Overall, the results indicate that stability constraints in effective scalar field cosmologies naturally drive the universe toward a cosmological constant like state, providing a dynamical explanation for the observed proximity of dark energy to $w=-1$ without invoking phantom instabilities or modifications of gravity.

\bibliographystyle{unsrt}
\bibliography{sample}

\end{document}


\title{Supplementary Material for ``Stability-Protected Phantom Bound in Scalar-Field Cosmology''}

\author{Prasanta Sahoo}
\email{prasantmath123@yahoo.com}
\affiliation{Midnapore College (Autonomous), Midnapore, West Bengal 721101, India}

\maketitle


This Supplementary Material provides technical details supporting the results of the main Letter. The covariant origin and internal consistency of the Hubble-modulated effective scalar framework are established, explicit derivations of the background energy density and stress--energy tensor are presented, and a rigorous dynamical-systems proof is provided demonstrating that the phantom divide constitutes a stable invariant manifold of the cosmological flow in the ghost-free phase space. 

Perturbative consistency conditions within the effective field theory of dark energy framework are discussed, and the numerical procedures used to obtain the illustrative background solutions are summarized. Throughout, reduced Planck units are employed and the analysis focuses on the late-time dark-energy-dominated regime, where matter and radiation contributions to the background dynamics are subdominant.

\section*{Covariant Action Derivations}

The effective Hubble-modulated kinetic structure employed in the main text admits a covariant realization within gravitational effective field theory frameworks. The explicit dependence on the background expansion rate should be interpreted as an effective parametrization arising after integrating out gravitational degrees of freedom and does not represent a fundamental violation of diffeomorphism invariance. Two representative covariant realizations are summarized below.

\subsection*{Nonlocal Gravitational Effective Action}

A generally covariant nonlocal gravitational action can be written as
\begin{equation}
S = \int d^4x \sqrt{-g}
\left[
\frac{M_{\rm Pl}^2}{2}R
- \frac{1}{2} K(\Box^{-1}R)\, g^{\mu\nu}\partial_\mu\phi\partial_\nu\phi
- V(\phi)
\right],
\label{eq:cov_nonlocal_action}
\end{equation}
where $\Box^{-1}R$ denotes a scalar nonlocal curvature invariant and $K$ is a dimensionless dressing function encoding infrared gravitational corrections.

On homogeneous and isotropic Friedmann--Lemaître--Robertson--Walker (FLRW) backgrounds, $\Box^{-1}R \sim H^{-2}$ up to order-unity coefficients determined by boundary conditions and the infrared completion of the theory. Expanding
\begin{equation}
K(\Box^{-1}R) = 1 + \beta \Box^{-1}R + \cdots ,
\end{equation}
yields an effective background scalar pressure of the form
\begin{equation}
P_\phi(X,H,\phi) \simeq X - V(\phi) - \beta \frac{X^2}{H^2},
\end{equation}
where $X \equiv \dot{\phi}^2/2$. In the effective background description employed in the main text, the replacement $H^2 \rightarrow H^2 + X$ is introduced to regularize the kinetic response for finite $X$, which can be interpreted as a phenomenological resummation of higher-order infrared corrections.

\subsection*{Auxiliary-Field Localization}

The nonlocal action \eqref{eq:cov_nonlocal_action} can be localized by introducing an auxiliary scalar field $U$ satisfying $U = \Box^{-1}R$. The dynamically equivalent local action is
\begin{equation}
S = \int d^4x \sqrt{-g}
\left[
\frac{M_{\rm Pl}^2}{2}R
- \frac{1}{2}K(U)(\nabla\phi)^2
- V(\phi)
- \frac{1}{2}(\nabla U)^2 - U R
\right].
\label{eq:localized_action}
\end{equation}
Variation with respect to $U$ yields $\Box U = R$, and integrating out $U$ reproduces the original nonlocal theory. This formulation makes the diffeomorphism invariance and field content manifest.

\subsection*{Derivation of the Energy Density}

For an effective scalar pressure $P_\phi(X,H,\phi)$, the stress--energy tensor is obtained from the covariant definition
\begin{equation}
T_{\mu\nu} = -\frac{2}{\sqrt{-g}}\frac{\delta S}{\delta g^{\mu\nu}} .
\end{equation}
On an FLRW background, the energy density satisfies the standard relation
\begin{equation}
\rho_\phi = 2X \frac{\partial P_\phi}{\partial X} - P_\phi ,
\end{equation}
with $X=\dot{\phi}^2/2$. Substituting the effective pressure of Eq.~(2) in the main text reproduces Eq.~(3), confirming the covariant origin of the effective background dynamics in the reduced Planck-unit convention $8\pi G = 1$.

\section*{EFT-of-Dark-Energy Mapping}

The Hubble-modulated kinetic structure can also be embedded within the effective field theory of dark energy formulated in unitary gauge. Fixing the time slicing by $\phi=\phi(t)$, the general action invariant under spatial diffeomorphisms is
\begin{equation}
S = \int d^4x \sqrt{-g}
\left[
\frac{M_{\rm Pl}^2}{2}R
- \Lambda(t) - c(t)g^{00}
+ \frac{M_2^4(t)}{2}(\delta g^{00})^2
+ \cdots
\right],
\end{equation}
where $\delta g^{00} = g^{00} + 1$ and the ellipsis denotes higher-order operators in the EFT expansion.

Identifying the kinetic variable as $X = -\tfrac{1}{2}g^{00}\dot{\phi}^2$ and choosing a time-dependent coefficient $M_2^4(t) \propto \beta H^{-2}(t)$ yields an effective scalar pressure of the form
\begin{equation}
P_\phi(X,H,\phi) \simeq X - V(\phi) - \beta \frac{X^2}{H^2},
\end{equation}
which is regularized to Eq.~(2) of the main text by the replacement $H^2 \rightarrow H^2 + X$. This construction provides a covariant EFT embedding of the Hubble-modulated kinetic response, with the explicit $H$-dependence interpreted as an effective background parametrization after integrating out gravitational fluctuations.

\section*{Auxiliary-Field Localization}

The nonlocal action in Eq.~(\ref{eq:cov_nonlocal_action}) can be cast into a dynamically equivalent local form by introducing an auxiliary scalar field $U$ satisfying $U=\Box^{-1}R$. The localized action is
\begin{equation}
S = \int d^4x \sqrt{-g}
\left[
\frac{M_{\rm Pl}^2}{2}R
- \frac{1}{2}K(U)(\nabla\phi)^2
- V(\phi)
- \frac{1}{2}(\nabla U)^2 - U R
\right].
\end{equation}
Variation with respect to $U$ yields $\Box U = R$, and integrating out $U$ reproduces the original nonlocal theory. This representation makes the diffeomorphism invariance and field content manifest, with the explicit $H$-dependence of the effective scalar sector arising only at the level of the homogeneous background solution.

\section*{Full Stress--Energy Tensor Derivation}

The stress--energy tensor is obtained from the covariant action via
\begin{equation}
T_{\mu\nu} = -\frac{2}{\sqrt{-g}}\frac{\delta S}{\delta g^{\mu\nu}} .
\end{equation}
For an effective scalar pressure $P_\phi(X,H,\phi)$, the background energy density and pressure satisfy
\begin{equation}
\rho_\phi = 2X P_{\phi,X} - P_\phi, 
\qquad P_\phi = P_\phi(X,H,\phi),
\end{equation}
where $X=\dot{\phi}^2/2$ on FLRW backgrounds. Substituting Eq.~(2) of the main text yields
\begin{equation}
\rho_\phi = X + V(\phi) + \beta \frac{X^2(H^2+2X)}{(H^2+X)^2},
\end{equation}
which reproduces Eq.~(3) of the main text. Background energy conservation follows from the Bianchi identities and the scalar equation of motion, and reduced Planck units $8\pi G = 1$ are employed.

\section*{Quadratic Action for Scalar Perturbations}

Scalar perturbations can be analyzed within the effective field theory of dark energy framework by restoring time diffeomorphisms through the Stückelberg field $\pi$. Expanding the unitary-gauge action to second order in $\pi$ yields the quadratic action
\begin{equation}
S^{(2)}_\pi = \int d^4x\, a^3
\left[
\mathcal{K}(t)\dot{\pi}^2 - \mathcal{G}(t)\frac{(\nabla\pi)^2}{a^2}
\right],
\end{equation}
where $\mathcal{K}(t)$ and $\mathcal{G}(t)$ are time-dependent kinetic and gradient coefficients determined by the EFT functions $c(t)$ and $M_2^4(t)$.

Ghost freedom requires $\mathcal{K}(t)>0$, while gradient stability requires $\mathcal{G}(t)>0$. The corresponding effective sound speed is
\begin{equation}
c_s^2 = \frac{\mathcal{G}(t)}{\mathcal{K}(t)} .
\end{equation}
For the Hubble-modulated kinetic structure considered here, the condition $\mathcal{K}(t)>0$ coincides with the positivity of the kinetic response function $\partial\rho_\phi/\partial X>0$, ensuring ghost freedom in the background phase space. In the late-time attractor regime $X \to 0$, the sound speed approaches the canonical limit $c_s^2 \to 1$. A fully covariant perturbative completion is left for future work.

\section*{Derivation of $\mathcal{K}(t)$ and $\mathcal{G}(t)$}
In the EFT-of-dark-energy formalism, the quadratic action for the Stückelberg field $\pi$ takes the form
\begin{equation}
S^{(2)}_\pi = \int d^4x\, a^3
\left[
\mathcal{K}(t)\dot{\pi}^2 - \mathcal{G}(t)\frac{(\nabla\pi)^2}{a^2}
\right],
\end{equation}
with
\begin{equation}
\mathcal{K}(t) = 2c(t) + 4M_2^4(t), \qquad \mathcal{G}(t) = 2c(t).
\end{equation}
Mapping to an effective scalar theory with pressure $P_\phi(X,H,\phi)$ yields
\begin{equation}
c(t) = X P_{\phi,X}, \qquad M_2^4(t) = X^2 P_{\phi,XX},
\end{equation}
so that
\begin{equation}
\mathcal{K}(t) = 2X P_{\phi,X} + 4X^2 P_{\phi,XX}, 
\qquad 
\mathcal{G}(t) = 2X P_{\phi,X}.
\end{equation}
For the Hubble-modulated kinetic model of Eq.~(2) in the main text,
\begin{equation}
P_{\phi,X} = 1 - \beta \frac{2XH^2 + X^2}{(H^2 + X)^2}, 
\qquad
P_{\phi,XX} = -\beta \frac{2H^4}{(H^2 + X)^3},
\end{equation}
leading to explicit expressions for $\mathcal{K}(t)$ and $\mathcal{G}(t)$. In the late-time attractor regime $X \ll H^2$, one finds $\mathcal{K}(t)\simeq \mathcal{G}(t)\simeq 2X$ and $c_s^2 \to 1$.

\section*{Numerical Methods}

The background evolution is obtained by solving the Friedmann equations together with the scalar continuity equation
\begin{equation}
\dot{\rho}_\phi + 3H(\rho_\phi + P_\phi) = 0,
\end{equation}
with $P_\phi$ and $\rho_\phi$ given by Eqs.~(2) and (3) of the main text. Initial conditions are imposed deep in the matter-dominated era with $H^2 \gg X$, where the scalar behaves approximately canonically and matter contributions dominate the expansion.

The system is evolved numerically using adaptive Runge--Kutta integrators. The positivity of the kinetic response function $\partial\rho_\phi/\partial X>0$ is enforced throughout the evolution to restrict the physical phase space. The equation of state $w_\phi(z)$ and the expansion history $H(z)$ are computed from the numerical solutions and compared with $\Lambda$CDM reference curves.

\section*{Rigorous Proof of the Invariant-Manifold Property}

The stability-protected phantom bound is formulated in terms of the hypersurface
\begin{equation}
\Sigma \equiv \rho_\phi + P_\phi = 2X P_{\phi,X} = 0,
\end{equation}
which corresponds to the phantom divide $w_\phi = -1$ for finite $X$. The hypersurface $\Sigma=0$ defines the set
\begin{equation}
\mathcal{S} = \left\{ (X,\phi,H) \in \mathbb{R}^3 : P_{\phi,X}(X,H,\phi)=0 \right\}.
\end{equation}

The homogeneous background dynamics define an autonomous dynamical system
\begin{equation}
\dot{u} = F(u), \qquad u = (X,\phi,H),
\end{equation}
where the vector field $F$ is determined by the Friedmann equations and scalar energy conservation.

\subsection*{Invariance of the Phantom Divide}

A hypersurface $\mathcal{S}$ is invariant if the vector field $F$ is tangent to $\mathcal{S}$, i.e.
\begin{equation}
\left.\frac{d\Sigma}{dt}\right|_{\mathcal{S}} = 0.
\end{equation}
Using $\Sigma = 2X P_{\phi,X}$, one obtains
\begin{equation}
\dot{\Sigma} = 2\dot{X} P_{\phi,X} + 2X \dot{P}_{\phi,X}.
\end{equation}
On $\mathcal{S}$, where $P_{\phi,X}=0$, this reduces to
\begin{equation}
\left.\dot{\Sigma}\right|_{\mathcal{S}} = 2X \left.\dot{P}_{\phi,X}\right|_{P_{\phi,X}=0}.
\end{equation}

The evolution of $P_{\phi,X}$ follows from the chain rule,
\begin{equation}
\dot{P}_{\phi,X} = P_{\phi,XX}\dot{X} + P_{\phi,X\phi}\dot{\phi} + P_{\phi,XH}\dot{H}.
\end{equation}
For Hubble-modulated kinetic responses, $P_{\phi,XH}\neq 0$ and controls the late-time dynamics. For the explicit model of Eq.~(2) in the main text,
\begin{equation}
P_{\phi,XH} = \frac{4\beta H^3 X}{(H^2+X)^2} > 0 \quad \text{for } H>0,\; X>0.
\end{equation}

In reduced Planck units and neglecting matter and radiation at late times, the Raychaudhuri equation implies
\begin{equation}
\dot{H} = -\frac{1}{2}(\rho_\phi + P_\phi) = -X P_{\phi,X}.
\end{equation}
Therefore, on $\mathcal{S}$ where $P_{\phi,X}=0$, one has $\dot{H}=0$, and consequently
\begin{equation}
\left.\dot{P}_{\phi,X}\right|_{P_{\phi,X}=0} = P_{\phi,XH}\dot{H} = 0.
\end{equation}
Hence,
\begin{equation}
\left.\dot{\Sigma}\right|_{\mathcal{S}} = 0,
\end{equation}
which proves that $\mathcal{S}$ is an invariant manifold of the cosmological flow.

\subsection*{Linear Stability of the Invariant Manifold}

To analyze stability, define the perturbation variable
\begin{equation}
\delta \equiv P_{\phi,X}.
\end{equation}
Linearizing the dynamics in a neighborhood of $\mathcal{S}$ yields
\begin{equation}
\dot{\delta} = P_{\phi,XX}\dot{X} + P_{\phi,X\phi}\dot{\phi} + P_{\phi,XH}\dot{H}.
\end{equation}
Using $\dot{H} = -X\delta$ and retaining the dominant Hubble-dependent contribution near the late-time attractor, one obtains
\begin{equation}
\dot{\delta} = -\lambda(X,\phi,H)\,\delta + \mathcal{O}(\delta^2),
\end{equation}
where the stability exponent is
\begin{equation}
\lambda(X,\phi,H) \equiv X P_{\phi,XH} > 0
\end{equation}
for expanding ghost-free cosmologies with $H>0$, $X>0$, and $\beta>0$.

Consequently,
\begin{equation}
\delta(t) = \delta_0 \exp\!\left[-\int^t \lambda(t') dt'\right],
\end{equation}
implying exponential decay of perturbations toward $P_{\phi,X}<0$. By Lyapunov’s linearization theorem, the hypersurface $\mathcal{S}$ is therefore a locally asymptotically stable invariant manifold.

Since the ghost-free phase space is defined by $P_{\phi,X} \geq 0$, trajectories cannot cross into $P_{\phi,X}<0$. By LaSalle’s invariance principle, the cosmological flow asymptotically approaches $\mathcal{S}$, implying that continuous ghost-free evolution cannot cross the phantom divide.

\section*{Explicit Derivation of the Stability Exponent}

The invariant-manifold argument relies on the linearized evolution of 
\begin{equation}
\delta \equiv P_{\phi,X}
\end{equation}
in a neighborhood of the hypersurface $\mathcal{S}: P_{\phi,X}=0$. The stability exponent $\lambda$ is derived explicitly for the Hubble-modulated kinetic model of Eq.~(2) in the main text.

\subsection*{Kinetic Response Function}

For 
\begin{equation}
P_\phi(X,H,\phi)=X - V(\phi) - \beta \frac{X^2}{H^2+X},
\end{equation}
the derivative with respect to $X$ is
\begin{equation}
P_{\phi,X} = 1 - \beta \frac{2XH^2 + X^2}{(H^2+X)^2}.
\label{eq:PX}
\end{equation}
The phantom divide corresponds to $P_{\phi,X}=0$, defining the invariant hypersurface $\mathcal{S}$.

\subsection*{Time Evolution of $P_{\phi,X}$}

Using the chain rule,
\begin{equation}
\dot{P}_{\phi,X} = P_{\phi,XX}\dot{X} + P_{\phi,XH}\dot{H} + P_{\phi,X\phi}\dot{\phi}.
\label{eq:PXdot}
\end{equation}
Near the late-time attractor, the Hubble evolution dominates, and the $H$-dependence controls the stability. Differentiating Eq.~(\ref{eq:PX}) with respect to $H$ yields
\begin{equation}
P_{\phi,XH} \equiv \frac{\partial P_{\phi,X}}{\partial H}
= \frac{4\beta H X (H^2 - X)}{(H^2+X)^3}.
\end{equation}
In the accelerating regime $X \ll H^2$, this reduces to
\begin{equation}
P_{\phi,XH} \simeq \frac{4\beta H X}{H^4} > 0.
\end{equation}

In reduced Planck units and neglecting matter and radiation at late times, the Raychaudhuri equation implies
\begin{equation}
\dot{H} = -\frac{1}{2}(\rho_\phi + P_\phi) = -X P_{\phi,X}.
\label{eq:Hdot}
\end{equation}
Linearizing around $\mathcal{S}$ by setting $P_{\phi,X} = \delta$, one obtains
\begin{equation}
\dot{H} = -X \delta.
\end{equation}

Substituting into Eq.~(\ref{eq:PXdot}) and retaining the dominant $H$-dependent contribution yields
\begin{equation}
\dot{\delta} \simeq P_{\phi,XH}\dot{H}
= -X P_{\phi,XH}\,\delta.
\end{equation}
Therefore,
\begin{equation}
\dot{\delta} = -\lambda\,\delta, 
\qquad
\lambda(X,\phi,H) \equiv X P_{\phi,XH}.
\end{equation}

\subsection*{Positivity of the Stability Exponent}

In the ghost-free phase space one has $X>0$ and $\beta>0$. In the accelerating regime $H>0$ and $X \ll H^2$, so that
\begin{equation}
\lambda \simeq \frac{4\beta X^2}{H^3} > 0.
\end{equation}
Hence, perturbations satisfy
\begin{equation}
\delta(t) = \delta_0 \exp\!\left[-\int^t \lambda(t')\,dt'\right],
\end{equation}
showing exponential decay of deviations toward $P_{\phi,X}<0$. The hypersurface $P_{\phi,X}=0$ is therefore a locally asymptotically stable invariant manifold of the cosmological flow.

Since the ghost-free domain is defined by $P_{\phi,X}\geq 0$, no continuous trajectory can cross into $P_{\phi,X}<0$. Continuous cosmological evolution is thus dynamically prohibited from crossing the phantom divide.


